\documentclass[10pt,twocolumn,twoside]{IEEEtran}

\IEEEoverridecommandlockouts                              

\usepackage{times} 
\usepackage{amsmath} 
\usepackage{amssymb}  
\usepackage{mathtools}
\usepackage{verbatim}
\usepackage{color}
\usepackage{bm}

\usepackage{dsfont}
\usepackage{xcolor}
\usepackage{tikz}
\usetikzlibrary{shapes.geometric, arrows, calc}
\usepackage{verbatim}
\usepackage{enumerate}
\usepackage{amsmath,amssymb,amsfonts}
\usepackage{amsmath,mathtools}
\usepackage{algorithmic}
\usepackage{graphicx}
\usepackage{textcomp}
\usepackage{xcolor}  
\usepackage{amsmath,bm}
\usepackage[font={footnotesize}]{caption}
\usepackage{cite}

\newtheorem{theorem}{Theorem}

\newtheorem{remark}{Remark}
\newtheorem{problem}{Problem}
\newtheorem{problem*}{Problem}

\newenvironment{proof}{\vspace{.1cm}\noindent{\sc
   Proof.}\hspace{0.10cm}\,\,}{$\hfill\Box$\vspace{.3cm}}

\newcommand{\rline}{{\mathbb R}}

\newcommand{\bbm}[1]{\left[\begin{matrix} #1 \end{matrix}\right]}
\newcommand{\sbm}[1]{\left[\begin{smallmatrix} #1
   \end{smallmatrix}\right]}

\newcommand{\rfb}[1]{\mbox{\rm
   (\ref{#1})}\ifx\undefined\stillediting\else:\fbox{$#1$}\fi}
   
\newcommand{\bluff}{{\hbox{\raise 15pt \hbox{\hskip 0.5pt}}}}

\newfont{\roma}{cmr10 scaled 1200}

\newtheorem{pro}{Proposition}

\newtheorem{assum}{Assumption}
 
\newtheorem{deff}{Definition}

\newtheorem{coro}{Corollary}

\title{\LARGE \bf
Dispersion Formation Control: from Geometry to Distribution
}

\author{Jin Chen$^{1}$, Jesus Bautista Villar$^2$, Bayu Jayawardhana$^3$ and Hector Garcia de Marina$^2$ 
\thanks{$^{1}$Jin Chen is with the Institute of Intelligent Vehicle, School of Mechanical Engineering, Shanghai Jiao Tong University, China. 
        {\tt\small chenjin920414@sjtu.edu.cn}}%
\thanks{$^{2}$Jesus Bautista Villar and Hector Garcia de Marina are with Department of Computer Engineer, Automation and Robotics, and with CITIC, the University of Granada, 
Spain.  {\tt\small hgdemarina@ugr.es, jesusbv@ugr.es}}
\thanks{$^{3}$ Bayu Jayawardhana is with the Engineering and Technology Institute Groningen, Faculty of Science and Engineering, University of Groningen, the Netherlands. 
        {\tt\small b.jayawardhana@rug.nl}}
}%

\begin{document}

\maketitle
\thispagestyle{empty}
\pagestyle{empty}

\begin{abstract}
	We introduce and develop the concept of \emph{dispersion formation control}, bridging a gap between shape-assembly studies in physics and biology and formation control theory. In current formation control studies, the control objectives typically focus on achieving desired local geometric properties, such as inter-agent distances, bearings, or relative positions. In contrast, our dispersion formation control approach enables agents to directly regulate the dispersion of their spatial distribution, a global variable associated with a covariance matrix. Specifically, we introduce the notion of \emph{covariance similarity} to define the target spatial dispersion of agents. Building on this framework, we propose two control strategies: a centralized approach to illustrate the key ideas, and a distributed approach that enables agents to control the global dispersion but using only local information. Our stability analysis demonstrates that both strategies ensure exponential convergence of the agents' distribution to the desired dispersion. Notably, controlling a global variable rather than multiple local ones enhances the resiliency of the system, particularly against malfunctioning agents. Simulations validate the effectiveness of the proposed dispersion formation control.
\end{abstract}

\section{Introduction}
In natural sciences, \emph{shape-assembly} pertains to a scientific field that investigates and characterizes the shape of multi-agent systems, with a focus on spatial distribution types, various moments, and the role of self-assembly in achieving specific shapes. Over the past decades, the problem of shape-assembly, where interacting agents form a particular shape through self-assembly, has instigated numerous studies across disciplines such as physics, biology, and engineering. In these studies, different \emph{artificial} confining potentials have been employed to predict spatial patterns of multi-agent systems and to facilitate the self-assembly of individual agents. 

In physics, the distribution of electrons in the Thomson problem \cite{von2012predicting} represents one of the most classic examples of shape-assembly studies. In biology, collective behaviors in self-assembling biological systems have gained significant attention in recent years \cite{vicsek2012collective}. In control theory and engineering, the shape-assembly problem is commonly called formation control \cite{oh2015survey} \cite{sun2023mean}. A key distinction between the notion of shape-assembly used in natural sciences and that in engineering science is that the former one uses natural laws for the network 
with typically complete graph assumption 
while the latter one focuses on practical implementations using non-complete graphs, which promote scalability by accounting for realistic sensing and communication constraints. In the literature on formation control, the desired shape is often defined by inter-agent geometrical constraints, with control problems focusing on designing the sum of local potential functions based on these constraints, as well as distributed control laws to asymptotically steer agents toward the desired minimum of such potentials. According to the different types of constraints and measurements, formation control can be classified into displacement- \cite{jadbabaie2003coordination}, distance- \cite{krick2009stabilisation,sun2016exponential,dimarogonas2008stability}, bearing-\cite{zhao2015bearing,Tri2019bearing,zelazo2015bearing} and angle-based \cite{chen2022simultaneous, chen2021maneuvering,chen2023angle} strategies, respectively. 

In the aforementioned literature, the desired formation has a precise description through the union of multiple local geometric variables, whereas the shape-assembly problem can focus more on a global description of the shape via its spatial distribution.
In various robotics applications, it is enough or even more beneficial for robots to form an ambiguous shape to achieve objectives such as going through a narrow corridor, i.e., without strict local geometrical constraints, the multi-agent system can have better performance in certain (sub)tasks due to the inherent flexibility. 
This motivates us to introduce the concept of \emph{dispersion formation control}, which focuses on the variance of a \emph{non-precise} collective shape and explores self-assembly strategies for individual agents. Dispersion formation control differs significantly from classical formation control methods, particularly in terms of controlled variables since in this paper we focus on an attribute of a distribution, as illustrated in Table \ref{table_example2}. Unlike traditional approaches that focus on controlling multiple local variables, our distribution approach targets the control of a few global variables; specifically, the eigenvalues of the covariance matrix that describes the agents' spatial dispersion. This shift moves the system from a traditional multi-robot framework toward a more swarm-like system, enhancing robustness and resilience against individual agent malfunctions. 

There exists a class of formation control methods known as density-based formation control, which also treats the desired formation shape as a spatial density distribution. However, our proposed dispersion formation control framework fundamentally differs from density-based approaches in terms of problem formulation, underlying assumptions, information availability, and controller design. In density-based formation control \cite{sinigaglia2022density,sinigaglia2025robust,bandyopadhyay2017probabilistic}, the task is typically formulated as an Optimal Control Problem (OCP) governed by partial differential equations (PDEs), where the objective is to minimize the discrepancy between the current agent density and the target density distribution over time. Solving this OCP yields a global vector field, which guides the motion of agents toward the desired formation. This approach generally requires a centralized computational unit that has access to global state information, including both the current and desired densities, in order to compute and transmit the local control inputs or vector fields for individual agents. In contrast, our dispersion formation control framework originates from conventional formation control rather than optimal control formulations. Each agent relies solely on local interactions, using relative position measurements and neighbor-to-neighbor communication to compute its control input. Additionally, the method does not require any knowledge of the current global distribution at any time. Instead, the desired formation shape emerges through distributed feedback control based entirely on local information exchange and errors of covariance matrix. It transforms our system into ordinary differential equations (ODEs) rather than partial differential equations (PDEs) in density-based formation control, significantly reducing the complexity of property analysis and enabling improved performance.


\begin{table*}
\caption{Distinctions among displacement-, distance-, bearing-, angle-based and dispersion formation control}
\label{table_example2}
\begin{center}
\begin{tabular}{|c||c||c||c||c||c||c|}
\hline
\quad &Displacement-\cite{jadbabaie2003coordination} &Distance-\cite{krick2009stabilisation,sun2016exponential,dimarogonas2008stability} &Bearing-\cite{zhao2015bearing,Tri2019bearing,zelazo2015bearing} &Angle-\cite{chen2020angle, chen2021maneuvering,chen2023angle} &Dispersion \\
\hline
Sensed variables &Relative positions &Relative positions &Relative bearings &Inter-edge angles &Relative positions\\
\hline
Controlled variables &Relative positions &Inter-agent distances  &Relative bearings  &Inter-edge angles &Covariance matrix eigenvalues\\
\hline
\end{tabular}
\end{center}
\end{table*}

In terms of paper organization, we begin by introducing the notation, graph theory, the definition of dispersion, and singular perturbation theory in Section II. To better illustrate the concepts behind our approach, we propose two dispersion control laws in Sections III and IV: a centralized control law using a complete graph and a distributed control law based on estimators for the formation's covariance matrix. We also present the exponential stability analysis of both systems in Sections III and IV. Finally, Section V provides numerical simulations validating our results.

\section{Preliminaries}

\subsection{Notation}
Let $\left| \mathcal{S} \right|$ denote the cardinality of a given set $\mathcal{S}$. For a vector $x \in \mathbb{R}^n$, $x^{\top}$ denotes the transpose of $x$ and the 2-norm of $x$ is denoted by $\left\| x \right\| =\sqrt{x^{\top}x}$. 
The $n \times n$ identity matrix is denoted by $I_n$. Additionally, for given matrices $A\in \mathbb{R}^{m\times n}$ and $B\in \mathbb{R}^{p\times q}$, the Kronecker product of $A$ and $B$ is denoted by $A\otimes B\in \mathbb{R}^{mp\times nq}$, and for a given dimension $d=2,3$, we denote $\overline{A}=A\otimes I_d\in \mathbb{R}^{md\times nd}$ where $d$ mostly refers to $d=\{2,3\}$ for the $\{2,3\}$D Cartesian space. We denote by $1_a$ or $0_a\in\mathbb{R}^a$ the column vector with dimension $a\in\mathbb{R}^+$. We denote that a square symmetric matrix is positive definite by $C\succ 0$ and semipositive definite by $C\succeq 0$.


\subsection{Graph theory}\label{subsec:graph}
An undirected graph $\mathcal{G}$ is defined by the tuple $\left( \mathcal{V},\mathcal{E} \right)$, where $\mathcal{V}=\left\{ 1,2,\cdots ,n \right\} $ is the \textit{vertex} set and $\mathcal{E}\subseteq \mathcal{V}\times \mathcal{V}$ is the \textit{edge} set with $\left| \mathcal{E} \right|$ number of edges. 
If $(i,j)\in\mathcal{E}$, the ordered pair $(i,j)$ refers to the edge whose direction is represented by an arrow with node $i$ as the tail and node $j$ as the head. We assume throughout the paper that there is no self-loop in the considered graph $\mathcal{G}$, i.e., $\left( i,i \right) \notin \mathcal{E}$ for all $i\in \mathcal{V}$. 
The set of neighbors of vertex $i$ is denoted by $\mathcal{N}_i\triangleq \left\{ j\in \mathcal{V}\,|\left( i,j \right) \in \mathcal{E} \right\}$. We define the \textit{incidence matrix} $B\in \mathbb{R}^{\left|\mathcal{V}\right| \times \left|\mathcal{E}\right|}$ of $\mathcal{G}$ as
\begin{equation} \label{eq:incidence_matrix}
b_{ik} \triangleq \begin{cases}
	+1 \quad \text{if} \quad i=\mathcal{E}^{\text{tail}}_k\\
	-1 \quad \text{if} \quad i=\mathcal{E}^{\text{head}}_k\\
	0 \quad \text{otherwise},
\end{cases}. 
\end{equation}

For an undirected graph, the \textit{Laplacian matrix} $L \in \mathbb{R}^{\left|\mathcal{V}\right| \times \left|\mathcal{V}\right|}$ \cite[Chapter 6]{bullo2020lectures} is given by $L = BB^\top$. It is known that when the graph $\mathcal{G}$ is connected, then the Laplacian matrix $L$ has a single eigenvalue equals zero \cite[Chapter 3]{bullo2020lectures}, whose associated eigenvector is $\mathbf{1}_N$ as $B^\top\mathbf{1}_N = 0$. 

\begin{figure}[h!]
    \centering 
    \includegraphics[trim={0cm 0cm 0cm 0cm}, clip, width=1\columnwidth]{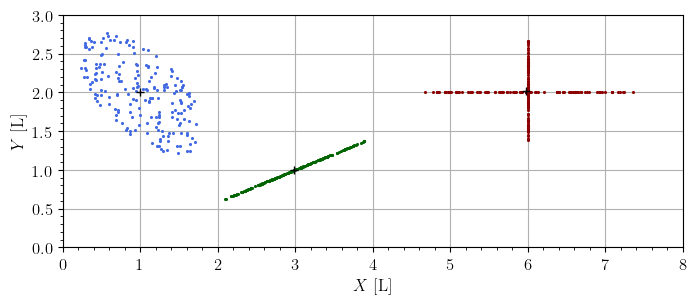}
    \caption{
        {
        This plot shows three frameworks with $n=200$ agents each one. The blue agents under a uniform distribution with eigenvalues $\lambda_1 =0.06$, $\lambda_2 = 0.26$ corresponding to the covariance matrix in \eqref{eq:convariance_matrix}. The red agents exhibit a cross-shaped distribution with the same eigenvalues as the blue ones; hence, $C_{\text{blue}} \sim C_{\text{red}}$. The green agents illustrate a degenerate configuration with $C\succeq 0$, where the distribution collapses to a line with $\lambda_1 =0$, $\lambda_2 = 0.25$.
        } 
    }
    \label{fig: cov}
\end{figure}

\subsection{Reference dispersion}\label{sec:ref_cloud_formation}

Given a finite collection of $n$ points $\left\{ p_i \right\} _{i=1}^{n}$ in $\mathbb{R}^d$ with $n\geqslant 2$ and $d \in \{2,3\}$ typically, a {\it configuration} is the set of agents' position, i.e., $p=\left[ p_{1}^{\top},\dots ,p_{n}^{\top} \right] ^{\top}\in \mathbb{R}^{dn}$, where $p_i \in \mathbb{R}^{d}$ is the position of $i$-th agent. A {\it framework} in $\mathbb{R}^d$ is a tuple $(\mathcal G,p)$, or  $\mathcal G(p)$, where each $i\in \mathcal{V}$ is mapped to $p_i$.

We note that the variance of a distribution captures how disperse is the cloud of points with respect to the average. For a given framework $\mathcal{G}(p)$, 
its covariance matrix $C \in \mathbb{R}^{d\times d}$ is defined by
\begin{equation}\label{eq:convariance_matrix}
    C = \frac{1}{N}\sum_{i=1}^{N} (p_i - p_c)(p_i - p_c)^\top
\end{equation}
where $p_c = \frac{1}{N} \sum_{i=1}^{N} p_i$ is the centroid of the multi-agent system, and, indeed we have that $C \succeq 0$, where the zero eigenvalues correspond to degenerate configurations, e.g., the deployment of agents form a straight line in the 2D plane. Note that the denominator in (\ref{eq:convariance_matrix}) is $\frac{1}{N}$, not $\frac{1}{N-1}$, because we are considering the entire set of points rather than a sample from a larger population. Using the Principal Component Analysis (PCA) \cite{al2021distributed}, we can compute the $d$ unit eigenvectors 
$\{\upsilon_j\}_{j=1,...,d}$, 
which correspond to the eigenvalues $\{\lambda_j\}_{j=1,...,d}$ with an increasing order of magnitude ($0<\lambda_1\leq \ldots \leq \lambda_d$). 
These eigenvectors represent the orthogonal directions of the agents' distribution in the $d$D-space where the corresponding eigenvalues are the variances along the direction of the eigenvectors.

\begin{deff}[Covariance Similarity]\label{Cov_equal}
    The frameworks $\mathcal{G}(p)$ and $\mathcal{G}({p}')$ are called {\it covariance similar} if and only if $C \sim {C}'$ where $\sim$ denotes matrix similarity and the matrices $C$ and ${C}'$ are the covariance matrices from $\mathcal{G}(p)$ and $\mathcal{G}({p}')$,  respectively.
\end{deff}

Since $C$ and ${C}'$ are real symmetric matrices, $C \sim {C}'$ if and only if $C$ and ${C}'$ have the same eigenvalues. We show in Figure \ref{fig: cov} two different frameworks that are covariance similar and one framework with a $C \succeq 0$. For a given arbitrary configuration of interest or reference configuration $p^* = \left[ {p_{1}^{*}}^{\top},\dots ,{p_{n}^{*}}^{\top} \right] ^{\top} \in \mathbb{R}^{nd}$, as an instance of the desired dispersion, 
we can define the reference covariance matrix by
\begin{equation}\label{eq:desired_covariance}
    C^* = \frac{1}{N}\sum_{i=1}^{N} (p_i^{*} - p^*_c)(p_i^{*} - p^*_c)^\top
\end{equation}
where $p^*_i \in \mathbb{R}^d$ is the position of $i$-th agent in the reference configuration and $p^*_c \in \mathbb{R}^d$ is the centroid point. Correspondingly, we define the \emph{desired dispersion} based on 
$C^*$ as follows.

\begin{deff}
\label{def:dispersion}
For a given $C^*\in\mathbb{R}^{d\times d}$, the configuration $p$ with the covariance matrix $C$ is at the \emph{desired dispersion} when it holds $C \sim C^*$
\end{deff}\vspace{0.1cm}

 The set of all configurations $p$ such that $C \sim C^*$ holds is denoted by $\mathcal S$ and is called the {\it set of the desired dispersion}. In this definition, the desired dispersion is parameterized by $C^*$, which is an invariant quantity that is satisfied by all admissible configurations $p$. Note that the set $\mathcal{S}$ 
corresponds to all frameworks which are covariance similar to the reference configuration $p^*$. Accordingly, the \textit{dispersion formation control} problem refers to the design of a controller such that the agents' position $p(t)$ converges to $\mathcal S$ 
as $t \rightarrow \infty$.

\begin{remark}
    Note that the dispersion imposes restrictions only on the covariance matrix $C$ instead of the number of agents. It implies that the desired covariance matrix can be defined by $N_1$ agents, while the desired dispersion can still be achieved with $N_2$ agents. This characteristic enhances the resiliency of the control law as it is allowed to "lose" or "gain" agents during the process.
\end{remark}




\subsection{Singular perturbations}
Later the agents will estimate the global covariance matrix and their centroid  distributively while they run the formation control law simultaneously in a slow-fast fashion. For the sake of convenience, please, allow us to recall the following result.
\begin{theorem}[{\cite[Theorem 11.4]{khalil2002nonlinear}}]\label{SF_theorem}
Consider the following singularly perturbed  systems
\begin{align}\label{fast_system}
	\dot x &= f(t,x,z,\varepsilon) \\
        \varepsilon \dot z &= g(t,x,z,\varepsilon), \label{slow_system}
    \end{align}
where $\varepsilon>0$ is a sufficient small constant, and the functions $f,g$ are smooth. 
Assume that the following conditions are satisfied for all 
    $(t,x,\varepsilon) \in \left[ 0, \infty \right) \times B_r \times \left[0, \varepsilon_0 \right]$:
\begin{itemize}
    \item $f(t,0,0,\varepsilon) = 0$ and $g(t,0,0,\varepsilon) = 0$;
    \item The equation
        $0 = g(t,x,z,0)$
    has an isolated root $z=h(t,x)$ such that $h(t,0) = 0$.
    \item Apply the change of variables by $y = z-h(t,x)$. The functions $f,g,h$, and their partial derivatives up to the second order are bounded for $y\in B_{\rho}$;
    \item The origin of the reduced system
        $\dot x = f(t,x,h(t,x),0)$;
    is exponentially stable;
    \item Assume $\tau = \frac{t}{\varepsilon}$. The origin of the boundary-layer system
    \begin{align*}
        \frac{dy}{d\tau} = g(t,x,y+h(t,x),0)
    \end{align*}is exponentially stable, uniformly in $(t,x)$.
\end{itemize}
\vspace{0.1cm}
Then, there exists $\varepsilon^*>0$ such that for all $\varepsilon<\varepsilon^*$, the origin of \eqref{fast_system}-\eqref{slow_system} is exponentially stable.        
\end{theorem}

\subsection{Problem formulation}

Consider the framework $\mathcal{G}(p)$ with $n$ agents in $\mathbb{R}^d$ ($n \geq 2$, $d \in \{2,3\}$), where the agent's position $p_i \in \mathbb{R}^d$ has dynamics
\begin{equation}\label{single-integrator}
    \dot p_i(t) = u_i(t),
\end{equation}
where $u_i(t) \in \mathbb{R}^d$ is the control input to be designed. Let $p_c$ and $C$ denote the centroid and covariance matrix of $\mathcal{G}(p)$, respectively. We denote by $u = \sbm{u_1^\top, \dots , u_n^\top}^\top \in \mathbb{R}^{dn}$ the stacked control input.

Using the notations defined earlier in Subsection \ref{sec:ref_cloud_formation}, consider $\mathcal{G}(p^*)$ as a reference framework in which we can define the reference covariance matrix $C^*$ and the associated set of desired disperse configurations $\mathcal S$. For defining the dispersion error between any position $p$ with its associated covariance matrix $C$ and the \emph{desired dispersion tuple} $(p^*,C^*)$, we can use the spectrum information $\lambda_1 \geq \ldots \geq \lambda_d>0$ and $\lambda^*_1\geq \ldots \geq \lambda^*_d>0$ of $C$ and $C^*$, respectively, as follows 
\begin{equation}\label{error_eigenvalue}
    \begin{aligned}
       e_\lambda(t) := \bbm{e_1(t), \cdots, e_d(t)}^\top  :=\bbm{\lambda_1(t) - \lambda^*_1,  \cdots,  \lambda_d(t) - \lambda^*_d}^\top. 
    \end{aligned}
\end{equation}
\begin{remark}
\label{rem: v1v2}
If we measure the covariance matrices (\ref{eq:convariance_matrix}) and (\ref{eq:desired_covariance}) in barycentric coordinates, i.e., $p_c = 0$, then our error signal \eqref{error_eigenvalue} is invariant up to rotations, since we control eigenvalues and not eigenvectors, and translations of the configuration $p$, which is coherent with Definition \ref{def:dispersion} of desired dispersion.
\end{remark}


\begin{problem}[\bf Dispersion formation control problem] For the multi-agent systems with dynamics \eqref{single-integrator} and for the given dispersion tuple $(p^*,C^*)$, design a control law $u$ such that 
$e_\lambda(t) \rightarrow 0$ as $t \rightarrow \infty$, i.e., $p(t) \rightarrow \mathcal{S}$ as $t \rightarrow \infty$.
\end{problem}

We will focus only on the 2D dispersion-formation problem since the extension to higher dimensions is straightforward.

\section{Dispersion control with a complete graph}


\subsection{Control law}

In this subsection, we study the required information for solving the aforementioned dispersion formation control problem. The sole-use of the local information by an individual concerning its neighbors, such as distances, relative positions or bearings, cannot be linked directly to global characteristics based on the overall distribution of the agents' positions. For simplifying the presentation, we firstly assume in this section that all agents have access to $C$ and the centroid position $p_c$, which can be measured collectively via a global sensor system such as an overhead camera for small robots. In the following section, each agent with only its local information will estimate this global information distributively. Based on the measured $C(t)$, each agent can compute the eigenvectors $\upsilon_1(t)$ and $\upsilon_2(t)$ associated to the eigenvalues $\lambda_1(t)$ and $\lambda_2(t)$. Further, we assume that the local coordinate among all agents is aligned, meaning agents are assumed to share an aligned coordinate frame (e.g., North \& East directions) while they can have different origin points. Note that this assumption is commonly accepted and widely adopted in various formation control strategies, such as bearing-based and angle-based formation control \cite{oh2015survey,zhao2015bearing,chen2022simultaneous}.
Let us define the relative position between agent $i$ and the configuration centroid with the barycentric coordinate $z_i^c :=\bbm{x_i^c \quad y_i^c}^\top = p_i - p_c$.

Using this information, we consider the following dispersion formation control law throughout this section:
\begin{equation}\label{eq:control_law_centrial}
    \dot{p}_i = u_i =-( e_1 \eta_1^i \upsilon_1 + e_2 \eta_2^i \upsilon_2) \quad i \in \mathcal{V},
\end{equation}
where $\eta_1^i = \left \langle z_i^c,\upsilon_1  \right \rangle$ and $\eta_2^i = \left \langle z_i^c,\upsilon_2  \right \rangle$ are the inner products that can be computed with the standard dot product in the Euclidean space, and  $e_1$ and $e_2$ are the dispersion errors in \eqref{error_eigenvalue}. 
Since the covariance matrix $C$ is real symmetric, $\upsilon_1$ and $\upsilon_2$ are orthogonal and unitary, or they can be chosen orthogonal in the case of being a multiple of the identity matrix. If we consider $\upsilon_1$ and $\upsilon_2$ as the basis of 2-D space, it follows by projecting on the eigenvectors the following expression \begin{equation}\label{eq:z_i^c}
\begin{aligned}
    z_i^c = \left \langle z_i^c,\upsilon_1  \right \rangle \upsilon_1 + \left \langle z_i^c,\upsilon_2  \right \rangle \upsilon_2  
    = \eta_1^i \upsilon_1 + \eta_2^i \upsilon_2.
\end{aligned}
\end{equation}
In the following, we will analyse the properties of the multi-agent system under the control law $u_i$ as in \eqref{eq:control_law_centrial}. 

\begin{pro}[Centroid invariance]\label{lemma:Centroid}
    For the multi-agent systems as in \eqref{eq:control_law_centrial}, the centroid position $p_c$ is invariant.  
\end{pro}
\begin{proof}
    It is straightforward to show that the dynamics of the centroid is given by 
\begin{equation}
\begin{aligned}
    \dot{p}_c &= \frac{1}{N} \sum_{i=1}^{N} (e_1 \left \langle z_i^c,\upsilon_1  \right \rangle \upsilon_1 + e_2 \left \langle z_i^c,\upsilon_2  \right \rangle \upsilon_2) \\
    &= \frac{1}{N}  \left(e_1 \left \langle \sum_{i=1}^{N} z_i^c,\upsilon_1  \right \rangle \upsilon_1 + e_2 \left \langle \sum_{i=1}^{N} z_i^c,\upsilon_2  \right \rangle \upsilon_2\right). 
\end{aligned}
\end{equation}
Since $\sum_{i=1}^{N} z_i^c = \sum_{i=1}^{N} (p_i -p_c) = 0$, we can conclude that $\dot p_c = 0$, i.e., the centroid $p_c$ is invariant.
\end{proof}

\begin{pro}[Eigenvector invariance]\label{lemma:Eigenvector}
    For the multi-agent system as in \eqref{eq:control_law_centrial}, the eigenvectors $\upsilon_1$ and $\upsilon_2$ are invariant. 
\end{pro}
\begin{proof}
    Let us analyze the invariance property of $\upsilon_1$. For $\upsilon_2$, the proof follows a similar fashion. Since $\upsilon_1(t)$ is the eigenvector corresponding to $\lambda_1(t)$ for $C(t)$, they satisfy 
    \begin{equation}\label{eigen_relationship_C}
    \begin{aligned}
        C(t)\upsilon_1(t) &= \lambda_1(t) \upsilon_1(t).
    \end{aligned}      
    \end{equation}
    Substituting \eqref{eq:convariance_matrix} and \eqref{eq:z_i^c} into \eqref{eigen_relationship_C}, we can obtain
    \begin{equation}
    \begin{aligned}
		\label{eq: 14}
        C\upsilon_1 &= \frac{1}{N}\sum_{i=1}^{N}(\eta_1^i \upsilon_1 + \eta_2^i \upsilon_2)(\eta_1^i \upsilon_1 + \eta_2^i \upsilon_2)^\top \upsilon_1 \\
        &= \frac{1}{N}\left(\sum_{i=1}^{N}{\eta_1^i}^2\right)\upsilon_1 + \frac{1}{N}\left(\sum_{i=1}^{N}\eta_1^i \eta_2^i\right)\upsilon_2 = \lambda_1 \upsilon_1,
    \end{aligned}
    \end{equation}
	recalling that $\upsilon_1^\top\upsilon_2 =0$ and $\upsilon_1^\top\upsilon_1 =1$, and (\ref{eq: 14}) is consistent with the definition of (co)variance: 
 \begin{equation}\label{eq:calculation_relationship}
         \frac{1}{N}\left(\sum_{i=1}^{N}{\eta_1^i}^2\right) = \lambda_1 \quad \text{and} \quad 
         \left(\sum_{i=1}^{N}\eta_1^i \eta_2^i\right) = 0.
     \end{equation}
    Let us now consider $C(t+\mathrm{dt})$ where $\mathrm{dt}$ is the infinitesimal increment of $t$. From the definition of time derivative, it follows that $C(t+\mathrm{dt}) = C(t) + \dot C(t) \mathrm{dt}$,
    where $\dot C(t) = \frac{1}{N}\sum_{i=1}^{N} \left( (\dot p_i - \dot p_c)(p_i - p_c)^\top +  (p_i - p_c)(\dot p_i - \dot p_c)^\top \right).$

    Since $\dot p_c = 0$, and using $\dot p_i$ as in \eqref{eq:control_law_centrial}, we have that
    \begin{equation}\label{eq:C(t+dt)}
    \begin{aligned}
         &C(t+\mathrm{dt})\upsilon_1 = C(t)\upsilon_1 + \dot C(t) \upsilon_1 \mathrm{dt} \\
         &= \lambda_1 \upsilon_1 - \frac{\mathrm{dt}}{N}\left(2\sum_{i=1}^{N}{\eta_1^i}^2e_1\upsilon_1 + \sum_{i=1}^{N}\eta_1^i \eta_2^i(e_1+e_2)\upsilon_2 \right) \\
         &= \lambda_1(1-2e_1\mathrm{dt})\upsilon_1.
    \end{aligned}\nonumber
    \end{equation}
    This equality implies that $\upsilon_1$ is the eigenvector of $C(t+\mathrm{dt})$ with eigenvalue of $\lambda_1(1-2e_1\mathrm{dt})$. In particular, it shows that $\upsilon_1$ remains the eigenvector of $C$ during the transitory motion. Hence, $\upsilon_1$ is invariant. 
\end{proof}


{\begin{remark}
    According to Proposition \ref{lemma:Eigenvector}, the eigenvectors $\upsilon_1$ and $\upsilon_2$ are invariant under the control law \eqref{eq:control_law_centrial}. If $\lambda_1(0) \neq \lambda_2(0)$, the eigenvectors $\upsilon_1$ and $\upsilon_2$ can be computed based on initial covariance matrix $C(0)$. 
    However, when $\lambda_1(0) = \lambda_2(0)$, 
    which represents the case of an uniform radial distribution of agents, 
    we can choose any two orthogonal unit vectors as initial conditions 
    for the control law \eqref{eq:control_law_centrial}. If $\lambda_1(T^*) = \lambda_2(T^*)$ for any $T^* > 0$, then $v_1$ and $v_2$ must be chosen as in $t=0$.

\end{remark}}


\begin{pro}[Collision Avoidance]\label{lemma:collision}
	There are no collisions between agents under the control law \eqref{eq:control_law_centrial} for all $t\geq 0$. 
\end{pro}
\begin{proof}
	For any given $\mathcal{G}(p)$, the centroid and eigenvectors of the formation with the proposed control law \eqref{eq:control_law_centrial} are invariant according to Proposition \ref{lemma:Centroid} and Proposition \ref{lemma:Eigenvector}. In fact, we are \emph{scaling} continuously the formation along the eigenvectors $v_{\{1,2\}}$. Therefore, for any given time $t\geq 0$, the position $p(t)$ satisfies $p(t) = \overline{T}(t)  p(0)$, where $\overline{T}(t) = T(t)\otimes I_n$ with $T(t) \in \mathbb{R}^{2 \times 2}$ being a non-zero diagonal scaling matrix in the basis $\{v_1,v_2\}$. It implies that $p(t)$ is an affine transformation of $p(0)$; hence, no collisions between agents for all $t \geq 0$.
\end{proof}

Note that if the desired \( C^* \) is positive semidefinite, then it is possible for two or more agents to converge to the same position due to the degeneracy of the desired distribution, for example, 2D agents asymptotically converging to the same point on a 1D \emph{desired line}. Also note that the scaling $\overline T$ along the eigenvectors $v_{\{1,2\}}$ also implies that there is a distribution-type invariance, e.g., if the initial distribution of agents is uniform or normal, it remains uniform or normal for all $t > 0$.

\subsection{Stability analysis}

In this subsection, we will analyze the stability of the multi-agent system with the dispersion formation control law in 
\eqref{eq:control_law_centrial}. For any given $e_\lambda \in \mathbb{R}^2$, let us define the sets $\mathcal{Q}_1, \mathcal{Q}_2$ and $\mathcal{Q}$ by $\mathcal{Q}_1 :=  \{e_\lambda : e_1 = -\lambda_1^*\}$, $\mathcal{Q}_2 := \{e_\lambda : e_2 = -\lambda_2^*\}$ and $\mathcal{Q} := \mathcal{Q}_1 \cup \mathcal{Q}_2$. 

\begin{theorem}[Almost Global Exponential Stability]\label{pro:almost_global}
    Consider the multi-agent system with the given dispersion formation control law \eqref{eq:control_law_centrial}. The trajectory $e_\lambda(t)$ of the closed-loop system converges asymptotically and exponentially to $e_\lambda = 0$ for any $e_\lambda(0) \in \mathbb{R}^2 \backslash \mathcal Q$ if and only if $\lambda^*_{\{1,2\}} > 0$. If at least one $\lambda^*_{\{1,2\}} = 0$ then $e_\lambda = 0$ is just asymptotically stable. 
\end{theorem}\vspace{0.1cm}

\begin{proof}
	Consider $\lambda^*_{\{1,2\}} > 0$, we have that $\lambda_1 = \frac{1}{N}\left(\sum_{i=1}^{N}{\eta_1^i}^2\right) = \frac{1}{N}\left(\sum_{i=1}^{N}{\left \langle z_i^c,\upsilon_1  \right \rangle}^2\right)$  based on the relation \eqref{eq:calculation_relationship}. Hence $\lambda_1 \geq 0$ and the case $\lambda_1 = 0$ holds if and only if $\eta_1^i = 0$ for all $i \in \mathcal{V}$. Based on \eqref{eq:control_law_centrial}, \eqref{eq:z_i^c}, Proposition \ref{lemma:Centroid} and Proposition \ref{lemma:Eigenvector}, the time-derivative of $\eta_1^i$ satisfies 
\begin{equation}
\begin{aligned}
    \dot \eta_1^i = \left \langle \dot z_i^c,\upsilon_1  \right \rangle + \left \langle  z_i^c,\dot \upsilon_1  \right \rangle = \left \langle \dot p_i,\upsilon_1  \right \rangle = -e_1\eta_1^i,
\end{aligned}
\end{equation}
where we have used the fact that $\dot{\upsilon}_1 = 0$ and $\langle \upsilon_1,\upsilon_2\rangle = 0$. 
Consequently, we can calculate the time-derivative of $\lambda_1$ as follows
\begin{equation}\label{eq:lamda_1_system}
    \begin{aligned}
        \dot \lambda_1 = -\frac{2e_1}{N} \left(\sum_{i=1}^{N}{\eta_1^i}^2\right) = -2e_1\lambda_1.
    \end{aligned}
\end{equation}
The equilibrium points of \eqref{eq:lamda_1_system} are $\lambda_1 = 0$ and $\lambda_1=\lambda_1^*$ (associated to $e_1=0$). It can be checked that the equilibrium point $\lambda_1=0$ is unstable by linearizing \eqref{eq:lamda_1_system} around $\lambda_1=0$
\begin{equation}\label{eq:variational_lam1da_1}
	\dot{\tilde\lambda}_1 = 2\lambda_1^*\tilde\lambda_1,
\end{equation}
which shows that the equilibrium point $\lambda_1=0$ is unstable. 
The same conclusion can be obtained for $\lambda_2$ in 
a similar fashion.

We will now proceed in analyzing the stability of the equilibrium points associated to $\lambda_1=\lambda_1^*$ and $\lambda_2=\lambda_2^*$. Note that, by the definition of $e_\lambda$ in \eqref{error_eigenvalue}, the $e_\lambda$-system is given by
\begin{equation}\label{eq:e_system}
	\dot e_i = -2e_i\lambda_i, \quad i\in\{1,2\}.
\end{equation}
Since $e_i = \lambda_i - \lambda_i^*$, the solution of \eqref{eq:e_system} is given by
\begin{equation}
    e_i(t) = \frac{\lambda_i^*e_i(0)}{e_i(0) + \lambda_i^* - e_i(0)\exp{(-2\lambda_i^*t)}} \exp{(-2\lambda_i^*t)}.
\end{equation}
Note that for $\lambda_i^*=0$ the previous two equilibrium points for each $i=\{1,2\}$ merge and the solution of \eqref{eq:e_system} is
\begin{equation}
\label{eq: eco}
e_i(t) = e_i(0)/(2e_i(0)t + 1),
\end{equation}
and since $e_i(0) \geq 0$ for $\lambda_i^*=0$, we can conclude that $e_\lambda = 0$ is almost-globally exponentially stable if and only if $\lambda^*_{\{1,2\}} > 0$, otherwise just asymptotically stable if and only if at least one $\lambda^*_{\{1,2\}} = 0$.
\end{proof}


\section{Distributed dispersion control}
According to the proposed control law in \eqref{eq:control_law_centrial}, each agent $i$ has to use its barycentric position, and the eigenvalues and eigenvectors from the covariance matrix. Because the barycentric coordinate and the covariance matrix are centralized information, the result presented in the previous section is only applicable for a complete graph. In this section, we will present how to obtain this information by a centroid estimator within the following interation graph.



\vspace{0.2cm}

\begin{assum}\label{graph_auumption}
    The graph is $\mathcal{G}$ in the framework $\mathcal{G}(p)$ is undirected and connected. 
\end{assum}



\subsection{Relative position to centroid estimation}
In this subsection, we introduce the distributed centroid estimator presented in \cite{jesus2024resilient} for solving the source seeking problem\cite{10886464}, which forms the foundation for estimating the covariance matrix distributively. Consider the framework $\mathcal{G}(p)$ with $n$ agents in $\mathbb{R}^d$ ($n \geq 2$ and $d \geq 2$), with $p_i \in \mathbb{R}^d$ denote the position of agent for all $i  \in \{ 1, \cdots n\}$ and with the underlying graph $\mathcal{G} = (\mathcal{V},\mathcal{E})$ being undirected and connected. We assume the local frame of each agent is aligned. 
In this case, the centroid of all $p_i$'s is given by $p_c := \frac{1}{n} \sum_{i=1}^{N} p_i$.
Let us introduce the dynamics of the estimation for the barycenter coordinate $\hat p_i\in\mathbb{R}^d$
\begin{equation}\label{eq:centroid_estimator}
    \dot {\hat p}_i(t) = -\sum_{j \in \mathcal{N}_i} \left ( (\hat{p}_i(t)-\hat{p}_j(t)) - (p_i - p_j) \right ). 
\end{equation}

\begin{coro}[\cite{jesus2024resilient}]\label{centroid_estimation_The}
	For all $i\in\mathcal V$, the estimator (\ref{eq:centroid_estimator}) makes the limit $\lim_{t\to\infty}\left( \hat{p}_i(t)-\frac{1}{n} \sum_i^n \hat{p}_i(0)\right) \to (p_i-p_c)$ and it happens exponentially fast, i.e., if $\sum_i^n \hat{p}_i(0) = 0_d$ then $\hat{p}_i(t)$ converges exponentially fast to the barycenter coordinate $(p_i - p_c)$.
\end{coro}

Now we will exploit (\ref{eq:centroid_estimator}) for the estimation of the covariance matrix distributively. Let us define the \emph{partial covariance matrix} for each agent $i$ by
\begin{multline} \label{eq:partial_CM}
    C_i := \sbm{ c_1^i& c_2^i\\ 
 c_2^i& c_3^i} := (p_i-p_c)(p_i-p_c)^\top.
\end{multline}
We define the vector $c_i=\sbm{c_{1}^i & c_{2}^i & c_{3}^i}^\top \in \rline^3$ for each agent $i$, where $c_{1}^i,  c_{2}^i, c_{3}^i$ are as in \eqref{eq:partial_CM}. Then, let us define the centroid of $C_i$ by 
$    \bar C := \frac{1}{n}\sum_1^n C_i := \sbm{\Bar{c}_1& \Bar{c}_2\\ 
 \Bar{c}_2& \Bar{c}_3}$. 
According to the definition of $C_i$ in \eqref{eq:partial_CM}, we know that
\begin{equation}\label{eq:C_centroid}
\begin{aligned}
    \bar C = \frac{1}{n}\sum_{i=1}^n (p_i - p_c)(p_i - p_c)^\top = C,
\end{aligned}
\end{equation}
where, indeed, $C$ is the covariance matrix (\ref{eq:convariance_matrix}) for $\mathcal{G}(p)$.

By denoting the centroid of $c_i$ as $\Bar{c} := \frac{1}{n} \sum_1^n c_i$, we know that $\Bar{c} = \bbm{\Bar{c}_1 & \Bar{c}_2 &\Bar{c}_3}^\top$ where $\Bar{c}_1$, $\Bar{c}_2$, $\Bar{c}_3$ are components of $\bar C$ as shown in \eqref{eq:C_centroid}. Since $C = \bar C$, $\Bar{c}_1$, $\Bar{c}_2$, $\Bar{c}_3$ are the components of the covariance matrix $C$ as well. Then, applying the centroid estimator \eqref{eq:centroid_estimator}, we propose the following covariance matrix estimators
\begin{equation}\label{eq:C_M_estimator}
    \dot {\hat{c}}_i(t) = -\sum_{j \in \mathcal{N}_i} \left ( (\hat{c}_i(t)-\hat{c}_j(t)) - (c_i - c_j) \right ), \forall i \in \mathcal{V}
\end{equation}
where $\hat c_i := \bbm{\hat c_{1}^i & \hat c_{2}^i & \hat c_{3}^i}^\top$. Let us denote $\hat{C}_i := \sbm{\hat c_1^i(t)& \hat c_2^i(t)\\
\hat c_2^i(t)& \hat c_3^i(t)}$.

\begin{pro}
	\label{cor: C}
    Suppose $\sum_i^N \hat{c}_i(0) = \mathbf{0}_3$, then by following \eqref{eq:C_M_estimator} we can conclude that $C = C_i - \lim_{t \rightarrow \infty} \hat{C}_i$
exponentially fast.
\end{pro}
\begin{proof}
According to Lemma \ref{centroid_estimation_The}, by setting up $\sum_i^N \hat{c}_i(0) = \mathbf{0}_3$, we can derive that $\hat{c}_i(t)$ converges exponentially fast to $c_i - \Bar{c}$ as $t \rightarrow \infty$. Since $c_i$ is known by agent $i$, we can calculate $\Bar{c}$ by $\Bar{c} = c_i - \lim_{t \rightarrow \infty}\hat{c}_i(t)$. Therefore, the covariance matrix can be computed distributively by $C = C_i - \lim_{t \rightarrow \infty}\hat{C}_i$.
\end{proof}

\tikzstyle{block} = [draw, fill=white, rectangle, 
    minimum height=3em, minimum width=3em, text width=2.5cm, align=center]
\tikzstyle{block_small} = [draw, fill=white, rectangle, 
    minimum height=3em, minimum width=3em, text width=2cm, align=center]
\tikzstyle{block_kin} = [draw, line width=0.5mm, fill=black!10, rectangle, 
    minimum height=3em, minimum width=6em, text width=3cm, align=center]
\tikzstyle{block_nei} = [draw, line width=0.5mm, fill=blue!10, rectangle, 
    minimum height=3em, minimum width=3em, minimum height=1cm, text width=3cm, align=center]
\tikzstyle{block_cas} = [draw, line width=0.5mm, fill=gray!10, rectangle, 
    minimum height=3em, minimum width=7cm, minimum height=1.5cm, text width=6cm, align=center]
\tikzstyle{block_sys} = [draw, dashed, line width=0.5mm, fill=white, rectangle, 
    minimum height=3em, minimum width=12.0cm, minimum height=2.0cm, text width=6cm, align=center]

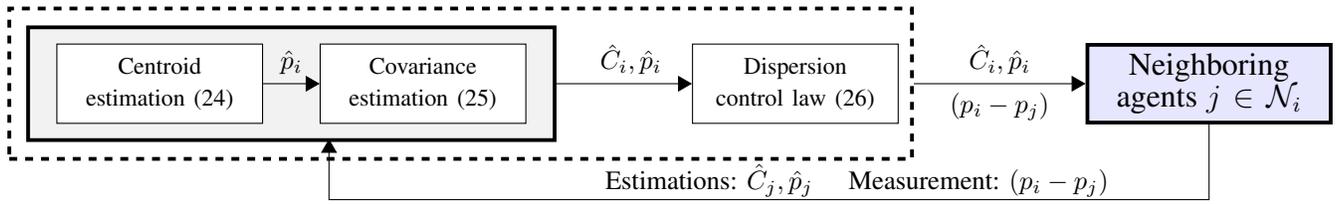
\begin{figure*}[t]
  \centering
  \begin{tikzpicture}[auto, node distance=2cm, >=triangle 60]
    \node [block_sys, shift={(4.5,0)}] (system) {};
    \node [block_cas, shift={(2.25,0)}] (cascade) {};
    \node [block, shift={(0.5,0)}] (centr) {{\small Centroid estimation \eqref{eq:controller_1}}};
    \node [block, right of=centr, node distance=3.5cm] (cov) 
            {\small Covariance estimation \eqref{eq:controller_2}};
    \node [block, right of=cov, node distance=4.95cm] (disp) 
            {\small Dispersion control law \eqref{eq:controller_3}};

    \node [block_nei, right of=disp, node distance=5.5cm] (nei) 
            {{\large Neighboring agents $j\in\mathcal{N}_i$}};
    
    \draw [->] (centr) -- node[] {$\hat p_i$} (cov);
    \draw [->] (cascade) -- node[pos=0.55] {$\hat C_i, \hat p_i$} (disp);
    \draw [->] (system)  --  node [] {$\hat C_i, \hat p_i$} node[below] {$(p_i - p_j)$} (nei);
    \draw [->] (nei.south)  -| ++(0,-1) -|  node [pos=0.2, yshift=0.6cm] {Estimations: $\hat C_j,\hat p_j$ \quad Measurement: $(p_i - p_j)$} ($(centr.south) + (2.25,-0.2)$);   
  \end{tikzpicture}
	\caption{Block diagram of the closed loop system \eqref{eq:controller_1}-\eqref{eq:controller_3} for all agent $i\in\mathcal{V}$. The \emph{superblock} with dashed line corresponds to the individual $i$'th agent, and it is replicated by its neighboring agents $j\in\mathcal{N}_i$. Each agent shares with its neighbors the estimations of its centroid $\hat p_i$ and its covariance $\hat C_i$. The motion of the agents due to the dispersion control law \eqref{eq:controller_3} makes the variation of the relative position $(p_i - p_j)$ that it is measured by the agents associated with the edge $(i,j)\in\mathcal{E}$. Note that the estimation of the centroid and covariance is in cascade (within the gray color block), and the estimators are in closed loop with the motion of the agents \eqref{eq:controller_3} due to the inputs measurement $(p_i-p_j)$ and their estimations $\hat p_j$ and $\hat C_j$.}
  \label{fig_control}
\end{figure*}


By defining $C_i - \hat{C}_i$, we conclude that $C_i - \hat{C}_i$ will converge exponentially fast to the covariance matrix $C$ as $t \rightarrow \infty$. Let us denote the eigenvalues of $C_i - \hat{C}_i$ as $\lambda_1^i$, $\lambda_2^i$ and the eigenvectors of $C_i - \hat{C}_i$ as $\upsilon_1^i$, $\upsilon_2^i$. Consequently, $\lambda_1^i$, $\lambda_2^i$, $\upsilon_1^i$, $\upsilon_2^i$ will respectively converge exponentially fast to $\lambda_1$, $\lambda_2$, $\upsilon_1$, $\upsilon_2$ as $t \rightarrow \infty$.

\begin{remark}
Note that neither (\ref{eq:centroid_estimator}) nor (\ref{eq:C_M_estimator}) depend on the total number of agents, nor does the control law (\ref{eq:control_law_centrial}) for reaching $C^*$. That is, the distributed integral solution has the potential to be resilient to variations in the total number of agents, since this information is not required.
\end{remark}

\subsection{Distributed control law and stability analysis}
In this subsection, we present our fully distributed control law so that we can control the dispersion of a multi-agent system. Our method can be separated in two steps that are in closed loop. Firstly, the agents use estimators \eqref{eq:centroid_estimator} and \eqref{eq:C_M_estimator} in cascade to distributively acquire their barycentric coordinate and the global covariance matrix. Subsequently the control law \eqref{eq:control_law_centrial} is applied using the estimated variables instead of centralized variables as inputs. Note that the motion of the agents influence both estimated values; hence, the closed loop system. However, the interconnection of the three algorithms do not happen with the same time rate but they need to be adjusted. In the following, these are the three algorithms that run at each agent $i\in\mathcal{V}$:
\begin{align}\label{eq:controller_1}
        \varepsilon_{\text{f}}\varepsilon_{\text{s}}\dot {\hat{p}}_i &= -\sum_{j \in \mathcal{N}_i} \left ( (\hat{p}_i-\hat{p}_j) - (p_i - p_j) \right ) \\
        \varepsilon_{\text{s}}\dot {\hat{C}}_i &= -\sum_{j \in \mathcal{N}_i} \left ( (\hat{C}_i-\hat{C}_j) -  (\hat{p}_i\hat{p}_i^\top - \hat{p}_j\hat{p}_j^\top) \right ) \label{eq:controller_2}\\ 
        \dot{p}_i &=-( e_1^i\left \langle z_i^c,\upsilon_1^i  \right \rangle\upsilon_1^i + e_2^i \left \langle z_i^c,\upsilon_2^i  \right \rangle \upsilon_2^i), \label{eq:controller_3}
\end{align}
where $\sum_i^N \hat{p}_i(0) = \mathbf{0}_2$, $\sum_i^N \hat{C}_i(0) = \mathbf{0}_{2 \times 2}$, $e_1^i = \lambda_1^i - \lambda_1^*$, $e_2^i = \lambda_2^i - \lambda_2^*$, $z_i^c = \hat{p}_i$, $\varepsilon_{\text{f}}, \varepsilon_{\text{s}} > 0$ are two sufficiently small constants to tune the slow-fast system, and $\lambda_1^i$, $\lambda_2^i$, $\upsilon_1^i$, $\upsilon_2^i$ are estimated eigenvalues and eigenvectors. As it is depicted in Figure \ref{fig_control}, we can identify that \eqref{eq:controller_1} estimates the barycentric coordinate of agent $i$ (the fastest dynamics) in cascade with \eqref{eq:controller_2} that estimates the global covariance matrix (the second fastest dynamics); and finally, in closed loop with the estimators we have the dynamics \eqref{eq:controller_3}, which is the motion that agent $i$ must follow (the slowest dynamics). This is a standard class of a slow-fast system as analyzed in  \cite{qin2021partial, maghenem2023singular}, and the exponential stability of \eqref{eq:controller_1} and \eqref{eq:controller_2} for the static case $\dot p = 0$, and the exponential stability of \eqref{eq:controller_3} for a desired dispersion $C^* \succ 0$ will assists us with seeking upperbounds for $\varepsilon_{\text{f}}$ and $\varepsilon_{\text{s}}$ {\cite[Theorem 11.4]{khalil2002nonlinear}}. In particular, note that since \eqref{eq:controller_1} and \eqref{eq:controller_2} are in cascade, the constant $\varepsilon_f$ is irrelevant for the analysis of the closed loop; however, it is convinient to set it small in practice to have a faster estimation of the centroid by \eqref{eq:controller_1} that can assist with a reliable estimation of the covariance matrix by \eqref{eq:controller_2}, i.e., the \emph{perturbation} to \eqref{eq:controller_2} will vanish much faster.

\begin{figure*}[h!]
    \centering 
    \includegraphics[trim={0cm 0cm 0cm 0cm}, clip, width=1.55\columnwidth]{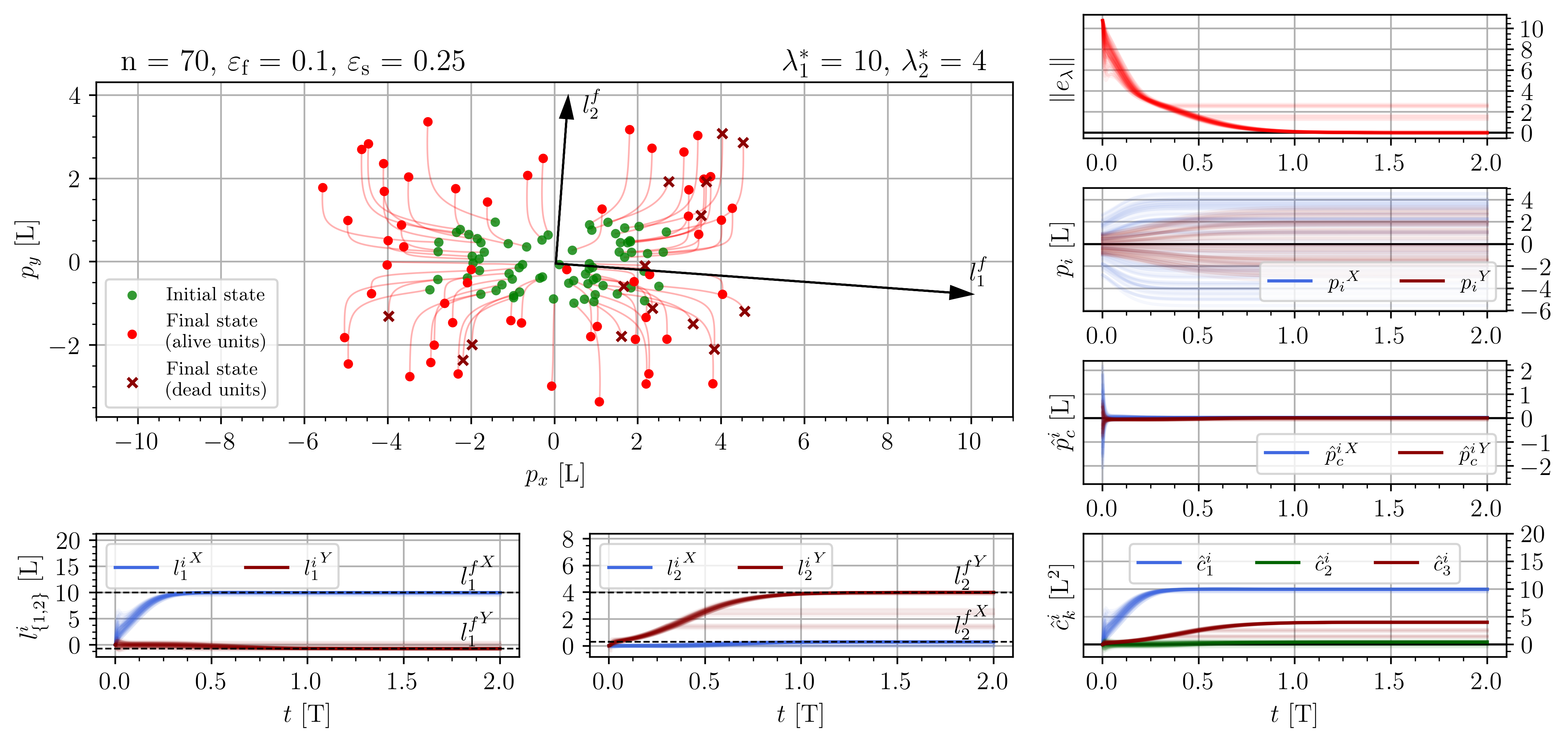}
    \caption{
        {
        The initial (green points) and final (red points) position of the agents are shown, along with their 2D trajectories (red lines). Black arrows represent the final eigenvalues and normalized eigenvectors of $C_i$ for $i=1$ as $l_{\{1,2\}}^f$, where $l_{\{1,2\}}^i = \lambda_{\{1,2\}}^i (v_{\{1,2\}}^i \, / \, \|v_{\{1,2\}}^i\|)$. The two plots below show how all the agents converge to such eigenvalues and eigenvectors.
        On the right, from top to bottom, the time evolution for each agent of $\|e_\lambda\|$, which converges exponentially fast to zero as $t \rightarrow \infty$; the agents' position $p_i$, demonstrating collision avoidance as their trajectories never intersect; and the estimated centroid and covariance matrix coefficients, based on $\varepsilon_{\text{f}} = 0.1$ and $\varepsilon_{\text{s}} = 0.25$. When an agent stops working (crosses), we keep its last value with a light red color.
        } 
    }
    \label{fig: distributed}
\end{figure*}

\begin{figure*}[h!]
    \centering 
    \includegraphics[trim={0cm 0cm 0cm 0cm}, clip, width=1.55\columnwidth]{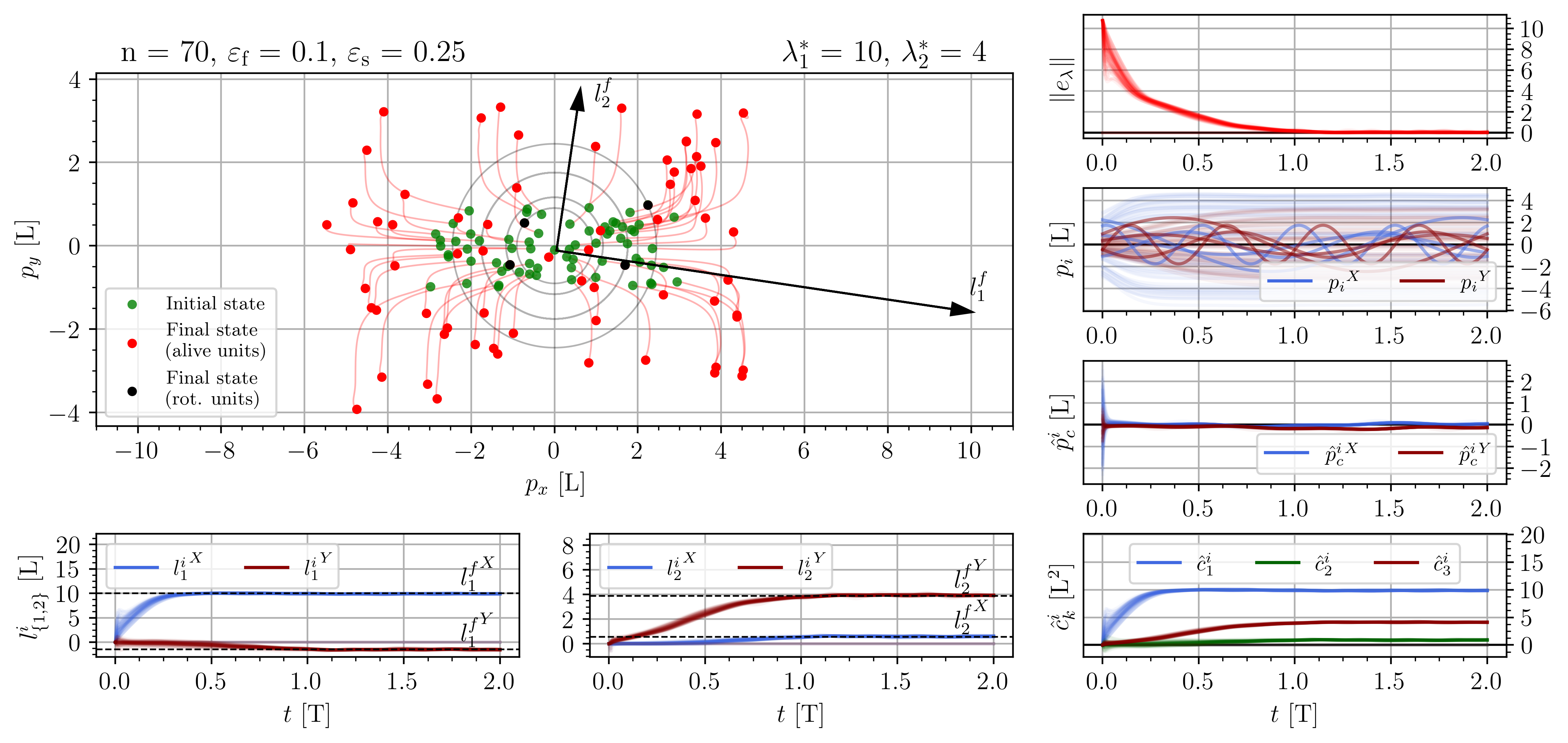}
    \caption{
        {
        Similar to Figure \ref{fig: distributed}, this figure shows the agents’ trajectories and convergence properties. However, in this case, 4 agents (black points) operate independently from the rest of the team, following circular trajectories with angular speed $\omega = \{4\pi, 2\pi, \frac{4}{3}\pi, \pi\}$ rad/s, while remaining connected to the communication graph. The right plot representing $p_i(t)$ shows the rotating agents with a larger alpha.
        } 
    }
    \label{fig: distributed_robustness}
\end{figure*}

\begin{theorem}
\label{thm: supermain}
 
	Consider a team of $n$ agents under a connected and undirected graph $\mathcal{G}$ with a desired dispersion given by $C^* \succ 0$ and the initial $C(0) \succ 0$. If $\varepsilon_{\text{f}}, \varepsilon_{\text{s}} > 0$ are sufficiently small, then the closed-loop \eqref{eq:controller_1}-\eqref{eq:controller_3} makes $\lim_{t\to\infty} C(t) \sim C^*$.
    
\end{theorem}
\begin{proof}

	We first know that for the static case, i.e., $\dot p_i = 0, \forall i\in\mathcal{V}$; hence, $p_i, p_c, C_i$ and $C$ are constant values, then we have that $\hat p_i(t) \to (p_i - p_c)$ and $\hat C_i(t) \to (C_i - C)$ exponentially fast, respectively, as $t\to\infty$ since \eqref{eq:controller_1} and \eqref{eq:controller_2} are in cascade, i.e, the output $\hat p_i$ of \eqref{eq:controller_1} goes as input for \eqref{eq:controller_2} but \eqref{eq:controller_1} does not need any output from \eqref{eq:controller_2}. Also note that the cascade system \eqref{eq:controller_1}-\eqref{eq:controller_2} is in closed-loop with \eqref{eq:controller_3}, i.e., $e^i_{\{1,2\}}, z_i^c$ and $v_{\{1,2\}}$ are inputs to \eqref{eq:controller_3} derived from the outputs of \eqref{eq:controller_1}-\eqref{eq:controller_2}, and $p_i$ as output of \eqref{eq:controller_3} goes as input of \eqref{eq:controller_1}-\eqref{eq:controller_2}.  Note that $\varepsilon_f$ does not play a role for the asymptotical stability for only the cascade system \eqref{eq:controller_1}-\eqref{eq:controller_2}; however, it is practically convenient to set it up small enough to guarantee exponential stability {\cite[Theorem 11.4]{khalil2002nonlinear}}, i.e., the \emph{perturbation} for \eqref{eq:controller_2} vanishes faster.

    
    We note that the desired configuration $q^* := \left(\hat p_i = (p_i - p_c), \hat C_i(t) = (C_i - C), e^i_{\{1,2\}} = 0\right)$ is an isolated equilibrium point of the closed-loop system \eqref{eq:controller_1}-\eqref{eq:controller_3}. When we consider systems \eqref{eq:controller_1}-\eqref{eq:controller_3} at the root $\hat p_i = (p_i - p_c), \hat C_i(t) = (C_i - C)$, the system \eqref{eq:controller_3} will degenerate into the reduced system as \eqref{eq:control_law_centrial} and the origin $e^i_{\{1,2\}} = 0$ is exponentially stable according to Theorem \ref{pro:almost_global}. In addition, the boundary-layer system takes the form 
    \begin{equation}
        \begin{aligned}
            \varepsilon_{\text{f}} \frac{\text{d}y_i}{\text{d}\tau} = &-\sum_{j \in \mathcal{N}_i} y_i - y_j \\
             \frac{\text{d}D_i}{\text{d}\tau} = &-\sum_{j \in \mathcal{N}_i} (D_i - D_j - y_iy_i^\top + y_jy_j^\top - (p_i - p_c)y_i^\top 
             \\ &-y_i(p_i - p_c)^\top + (p_j - p_c)y_j^\top +y_j(p_j - p_c)^\top)
        \end{aligned}
    \end{equation}
    where $y_i = \hat p_i - (p_i - p_c)$,  $D_i = \hat C_i - (C_i - C)$ and the origin is also exponentially stable based on {\cite[Theorem 11.4]{khalil2002nonlinear}}. Therefore, the exponential stability conditions required by Theorem~\ref{SF_theorem} for both the reduced and boundary-layer systems are satisfied.  Consequently, there exist sufficiently small constants $\varepsilon_{\{\text{s}, \text{f}\}} > 0$ such that $q^*$ is a (exponentially) stable equilibrium of the closed-loop system \eqref{eq:controller_1}-\eqref{eq:controller_3}.
    
\end{proof}

Roughly speaking, we note that both \eqref{eq:controller_1} and \eqref{eq:controller_2} are first-order consensus algorithms; therefore, if the signal to be estimated is dynamic and bounded, the estimation error is also bounded depending on how fast the signal to be tracked varies \cite{kia2019tutorial}. Consequently, we have this \emph{spiral} of the smaller $\|\dot p_i(t) \|$, the smaller $\|\hat p_i(t) - (p_i(t) - p_c(t))\|$ and $\|\hat C_i(t) - (C_i(t) - C(t))\|$ as $t\to\infty$, for example, once the transitory dominated by the initial conditions has almost vanished, and that again makes much smaller $\|\dot p_i(t) \|$. In particular, we can identify the dynamics of the estimation errors derived from \eqref{eq:controller_1}-\eqref{eq:controller_2} as the fast system \eqref{slow_system} and the error system derived from \eqref{eq:controller_3} as the slow system \eqref{fast_system}.

Note that for the desired $C^* \succeq 0$ we do not have exponential convergence for the errors in \eqref{eq:controller_3} because of (\ref{eq: eco}) making the stability proof more complex. Intuitively and roughly speaking, the result of Theorem \ref{thm: supermain} can be seen as the errors norms $\|e^i_{\{1,2\}}(t)\|$ in \eqref{eq:controller_3} go smaller since their original exponential convergence to zero, then $\|\dot p_i(t)\|$ goes smaller and therefore the estimation errors of \eqref{eq:controller_1}-\eqref{eq:controller_2} become also smaller since the dynamic signal that they are tracking goes slower, making all the closed-loop system converge to their desired configuration.

\begin{remark}
    Our proposed method is also applicable to multi-agent systems with directed graphs as long as the centroid and covariance matrix estimation are done, for example, when the network is balanced \cite{bullo2018lectures}.
\end{remark}

{

\section{Numerical validation}

In this section, we present the numerical simulation shown in Figure \ref{fig: distributed} to validate Theorem \ref{thm: supermain}. The simulation begins with an initial configuration of $n=70$ agents uniformly distributed over the region $\left[-3, 3\right]\times \left[-1, 1\right]$. The desired eigenvalues of the covariance matrix are set to $\lambda_1^* = 10$ and $\lambda_2^* = 4$. The underlying undirected communication graph has an algebraic connectivity of $\lambda_0 = 1.39$. In order to show the resiliency of the algorithm,  groups of five agents \emph{die} at $t = 0.3$ T, $0.5$ T, and $1$ T, respectively, while maintaining the connectivity of the communication graph throughout. The simulation results demonstrate that setting $\varepsilon_{\text{f}} = 0.1$ and $\varepsilon_{\text{s}} = 0.25$ is sufficient for exponential convergence to the target eigenvalues, thus providing numerical validation of Theorem \ref{thm: supermain}. Furthermore, since the initial eigenvalues are smaller than the desired ones, the agents must increase the dispersion of the formation, effectively increasing all inter-agent distances. This behavior confirms the absence of collisions, in accordance with Proposition \ref{lemma:collision}. 

Moreover, to evaluate the robustness of the proposed dispersion formation control law in the presence of non-cooperative agents, we design the simulation in Figure \ref{fig: distributed_robustness} using the same basic setup as in Figure \ref{fig: distributed}. In contrast, four robots in this scenario are not subject to control, i.e., they contribute for the estimation of the centroid and covariance although they do not follow \eqref{eq:controller_3} but orbit the origin with different angular velocities. The simulation results demonstrate the robustness of our method, showing that the multi-agent system can still get close enough to the desired dispersion formation despite the \emph{disturbances} introduced by non-cooperative agents.}

\section{Conclusions}
In this paper, we have introduced the concept of \emph{dispersion formation control}, which extends traditional formation control from geometry to distribution, allowing for the modeling of \emph{non-strict} collective shapes. To achieve the desired dispersion, we propose two control strategies: centralized and decentralized. For the centralized strategy, we demonstrate the exponential convergence of the multi-agent system while ensuring several practical properties, such as centroid invariance and collision avoidance. Additionally, by introducing a centroid estimator, we successfully extend the centralized control law into a distributed framework and establish the exponential stability of this approach using singular perturbation theory with slow-fast dynamics. The simulation results provide numerical validation for the proposed strategies, and show the robustness of the \emph{swarm} against non-cooperative agents.




\bibliographystyle{IEEEtran}
\bibliography{ref}

\end{document}